\documentclass[10pt]{article}
\usepackage{amsmath}
\usepackage{amssymb}
\usepackage{amsthm}
\pdfoutput=1

\usepackage[utf8]{inputenc}
\usepackage[T1]{fontenc}
\usepackage{microtype}

\usepackage{fourier}
\usepackage{caption}
\usepackage{mathrsfs}
\usepackage[colorlinks=true, pdfstartview=FitV, linkcolor=blue, citecolor=blue, urlcolor=blue]{hyperref}
\theoremstyle{plain}
\newtheorem{proposition}{Proposition}[section]

\newtheorem{theorem}[proposition]{Theorem}
\theoremstyle{definition}
\newtheorem{definition}[proposition]{Definition}

\newtheorem{remark}[proposition]{Remark}

\usepackage[usenames,dvipsnames]{xcolor}
\usepackage{url}
\usepackage{tikz}
\usepackage{pgfplots}
\pgfplotsset{compat=newest,ticks=none}
\usetikzlibrary{arrows,calc}
\usepackage{verbatim}
\usetikzlibrary{%
    decorations.pathreplacing,%
    decorations.pathmorphing%
}

\DeclareMathOperator{\sgn}{sgn}

\usepackage{amsmath}
\usepackage{amsthm}

\newcommand{\R}{\mathbf{R}}

\newcommand{\abs}[1]{\left\lvert #1 \right\rvert}
\newcommand{\avg}[1]{\bigl\langle #1 \bigr\rangle}
\newcommand{\jump}[1]{\bigl[ #1 \bigr]}

\usepackage{mathscinet}
\usepackage{cite}

\usepackage{microtype}
\begin{document}

\title{Lax Integrability of the  Modified Camassa-Holm Equation and the Concept of Peakons \let\thefootnote\relax\footnotetext{Published in J. Nonlinear Math. Phys. 23(4):563--572, 2016; a minor correction is made in this version.}}

\author{Xiang-Ke Chang\thanks{LSEC, Institute of Computational Mathematics and Scientific Engineering Computing, AMSS, Chinese Academy of Sciences, P.O.Box 2719, Beijing 100190, PR China; Department of Mathematics and Statistics, University of Saskatchewan, 106 Wiggins Road, Saskatoon, Saskatchewan, S7N 5E6, Canada; changxk@lsec.cc.ac.cn}
\and
  Jacek Szmigielski\thanks{Department of Mathematics and Statistics, University of Saskatchewan, 106 Wiggins Road, Saskatoon, Saskatchewan, S7N 5E6, Canada; szmigiel@math.usask.ca}
}
\date{}

\maketitle

\begin{abstract}
  In this Letter we propose that 
  for Lax integrable nonlinear partial differential equations the natural concept of 
  weak solutions is implied by the compatibility condition for the respective 
  distributional Lax pairs.  We illustrate our proposal 
  by comparing two concepts of weak solutions 
  of the modified Camassa-Holm equation pointing out that in the 
  \emph{peakon} sector (a family of non-smooth solitons)
only one of them, namely the one obtained from 
  the distributional compatibility condition, supports the 
  time invariance of the Sobolev $H^1$ norm.  
  \end{abstract}






\section{Introduction} 
\vspace{0.5 cm} 

The partial differential equation with cubic nonlinearity 
\begin{equation}\label{eq:m1CH}
m_t+\left((u^2-u_x^2) m)\right)_x=0, \qquad 
m=u-u_{xx}, 
\end{equation}
is a modification of the Camassa-Holm equation (CH) 
~$m_t+u m_x +2u_x m=0, \, ~m=u-u_{xx}$ \cite{CH}, for the shallow water waves.  
Originally, \eqref{eq:m1CH} appeared in the papers of Fokas \cite{fokas1995korteweg}, Fuchssteiner \cite{fuchssteiner1996some},  Olver and Rosenau\cite{olver1996tri} and was, later, rediscovered by Qiao \cite{qiao2006new,qiao2007new}.  
According to  the recent work \cite{kang-liu-olver-qu} this equation has a number of potential applications and features that make this equation worth studying: 
\begin{enumerate}
\item 
it models the uni-directional propagation of shallow water waves over a flat bottom 
(in \cite{kang-liu-olver-qu} the authors credit Fokas \cite{fokas-physically-important, fokas1995korteweg} for the derivation); 
\item it arrises from an intrinsic (arc-length preserving) invariant planar flow in 
Euclidean \\geometry ~\cite{gui2013wave}; 
\item it posses interesting non-smooth solutions such as non-smooth solitons (peakons).  
\end{enumerate}

As to the name of \eqref{eq:m1CH} we note that the derivation of this equation in \cite{olver1996tri} followed from an elegant but mysterious method of tri-Hamiltonian duality applied to the bi-Hamiltonian representation of the modified Korteweg-de Vries equation.  Since the CH equation can be obtained from the Korteweg-de Vries equation
by the same tri-Hamiltonian duality, it is therefore natural to refer to \eqref{eq:m1CH} as the modified CH equation (mCH) (\cite{gui2013wave,liu-liu-qu-multip}).  To avoid confusions we point out that 
some authors use the name FORQ to denote 
\eqref{eq:m1CH} (e.g. \cite{Himonas4peak}, \cite{himonas2014cauchy}).  
However, some ambiguity remains as there are other--distinct from 
\eqref{eq:m1CH}--equations also 
called {\bf the} modified CH equation, obtained from the CH equation on the basis of 
other considerations (see e.g. \cite{reyes-gorka}).

In the last 23 years since the appearance of \cite{CH} we have witnessed a sharp increase in reported "peakon equations":
the first significant addition from this class was the Hunter-Saxton (HS) equation 
\cite{hunter1991dynamics, hunter-zhang, bss-hunter}, followed by 
the Degasperis-Procesi (DP)
equation
\cite{degasperis1999asymptotic, degasperis2002new}, subsequently followed by  V. Novikov 
equation \cite{novikov2009generalisations, hone2008integrable} and then the Geng-Xue 
equation \cite{geng2009extension}.  Around the same time, 
mainly through the work of Z.Qiao and his collaborators \cite{qiao2006new, qiao2007new, qiao2008m, 
qiao2012integrable,song2011new}, with important contributions by others \cite{ivanov2006extended, fontanelli2006three, chen2006two, geng2011three, holm2011two, li2014four}, it 
became clear that the \textit{land of peakons} is vast.  
This leads to a natural question of understanding and possibly classifying peakon equations using some, as yet unknown, 
fundamental principles.

It should be pointed out that the quest to understand the 
general principles behind the peakon equations is not motivated solely by purely mathematical 
interest.  These equations show some remarkable features among which 
the stability of its peakon solutions is quite pertinent to this paper; we recall 
that 
the (orbital) stability of the CH peakons was established in \cite{constantin-strauss}, then the results for the periodic CH
peakons \cite{lenells-stabilitypp} and the DP peakons \cite{lin-liu-stabDPpeakons} followed suit. Another reason for studying the peakon equations 
is their potential hydrodynamical significance \cite{Constantin-Lannes}, 
for example in capturing essential features of 
waves of greatest height \cite{constantin-stokes,constantin-stokes-extreme}.

Before a classification of the peakon equations becomes viable it is 
important to make the following provisional statement: there are two types of 
known peakon equations, one type is in one way or the other described by Lax 
pairs, and equations in this class are expected to be integrable in some sense. 
The other type consists of non-integrable equations and the 
prototypical example was first introduced in the work of Degasperis, Holm and Hone \cite{dhh2} who defined a family of equations, later dubbed the $b$-equations,   
\begin{equation}\label{eq:b-equation}
u_t-u_{xxt}+(b+1)uu_x=bu_xu_{xx}+uu_{xxx}, \qquad b\in \R, 
\end{equation}
which reduce to CH and DP equations for special values $b=2$, $b=3$ respectively, 
while for other values of $b$ are non-integrable, failing the integrability test
discussed in \cite{degasperis1999asymptotic}; non-integrability later confirmed by more general methods by Mikhailov and 
Novikov in \cite{mikhailov-novikov-perturbative}.  The main point of 
\cite{dhh2} was that regardless of the value of $b$, \eqref{eq:b-equation} has peakon solutions obtained from 
the peakon ansatz $u=\sum_{j=1}^n m_j (t)e^{-\abs{x-x_j(t)}}$, 
by constraining smooth coefficients $m_j(t), x_j(t)$ to satisfy a system of 
ODEs.   

At this point our interest is in understanding the mathematical 
underpinnings of Lax integrable peakon equations, leaving non-integrable 
peakon equations for future work.  In our opinion two 
features of \eqref{eq:m1CH}  deserve to be further studied 
to inform future discussions of Lax integrable peakon equations: 
\begin{enumerate} 
\item as we argue below, using \eqref{eq:m1CH} as an example,  
the definition of what one means by weak solutions of a Lax integrable peakon equation is determined by the distributional character of its Lax pair; 
\item the Lax integrable peakon ODEs (see \eqref{eq:xmODE}) are isospectral 
deformations of an  \textit{oscillatory}
system in the sense of Gantmacher and Krein \cite{gantmacher-krein}.  
\end{enumerate} 
It is fair to say that these two items have very different ontological status.  
The question as to what one means by weak solutions in the case of an obvious 
lack of sufficient classical smoothness is a necessary starting point for any 
serious discussion of a PDE to which one is seeking generalized (non classical) 
solutions.  So without proper understanding of this issue it will be hard to have any 
meaningful discussion of what peakon equations are.  By contrast, the fact that the peakon sector of the DP equation comes from an oscillatory system 
was first proven in \cite{ls-cubicstring}, later confirmed for V. Novikow peakons 
in \cite{hls}, for the Geng-Xue equation in \cite{lundmark2014inverse}, while 
for the CH equation it is implicit in \cite{bss-stieltjes, bss-moment} once 
one realizes that the CH equation can be viewed as an isospectral deformation of the inhomogeneous 
string as defined by M. G. Krein in his study of oscillatory systems \cite{gantmacher-krein}.  So the second item is just a confirmation that the mysterious 
link between peakons and oscillatory systems persists for \eqref{eq:m1CH} and 
we have no additional insight into this issue at the moment, so we restrict ourselves just to mentioning this fact.   The remainder of the present letter is confined 
entirely to the  first item on the list above. 
\section{Weak solutions of the modified CH equation; a comparison}
 We note that equation 
\eqref{eq:m1CH} is the first peakon equation known to us for which the need to view 
the weak sector as a consequence of the distributional compatibility of its Lax pair  - rather than by defining weak solutions \textit{ a priori}, basing the definition  on the type of nonlinearity - is placed front and center.  

We recall that \eqref{eq:m1CH} can be written in bi-Hamiltonian 
form \cite{olver1996tri} 
\begin{equation}
    m_t=\mathcal{K} \frac{\delta \mathcal{H}_1}{\delta m}=\mathcal{J} \frac{\delta \mathcal{H}_2}{\delta m}, \qquad m=u-u_{xx}, 
\end{equation}
where the compatible Hamiltonian operators are: 
\begin{equation}
\mathcal{K}=-\partial _x m \partial_x^{-1} m \partial_x, \qquad \mathcal{J}=(\partial_x^3-\partial_x), 
\end{equation} 
with Hamiltonian functionals 
\begin{equation}
    \mathcal{H}_1=\int m u\,  dx, \qquad \mathcal{H}_2=\frac14 \int \big(u^4+2u^2u_x^2-\frac13 u_x^4\big) dx. 
\end{equation} 
Even though these are formal expressions it is not difficult to make rigorous analytic 
sense of them.  For example, if $u\in \mathcal{S}(\R)$ (Schwartz class of smooth, 
rapidly vanishing functions) then all expressions are well defined, in particular 
$\partial _x^{-1}$ can be taken to act as $\int_{-\infty}^x$ on any 
function $f\in \mathcal{S}(\R)$, giving after one integration by parts
\begin{equation} \label{eq:H1norm}
    \mathcal{H}_1=\int_{\R} (u^2+u_x^2)\, dx= ||u||_{H^1}^2.  
\end{equation} 
In other words the Hamiltonian $\mathcal{H}_1$ is the square of the 
Sobolev norm.  Since Hamiltonians are constants of motion we see that 
at least in the smooth sector of solutions the Sobolev norm, $||u||_{H^1}$, is a constant of 
motion.  The objective now is to extend \eqref{eq:m1CH} to functions 
with less smoothness, preserving as much of the original structure as possible.  
We recall from \cite{gui2013wave} one possible definition of weak solutions.   
\begin{definition} [Weak solution] \label{def:weakgui}
Given initial data $u_0$ in the Sobolev space $W^{1,3}(\R)$, the function 
$u\in L_{loc}^{\infty}([0,T], W_{loc}^{1,3}(\R))$ is said to be a weak solution to 
\eqref{eq:m1CH} with initial condition $u_0$ if the following identity holds
\begin{equation}
    \int_0^T\int_{\R} \big[ u\phi_t+\frac13 u^3 \phi_x+\frac13 u_x^3 \phi+ 
    p\star (\frac23 u^3+uu_x^2) \phi_x -\frac13(p\star u_x^3) \phi \big]\, dx+ 
    \int_{\R} u_0(x) \phi(0, x)\, dx=0
\end{equation} 
for all test functions $\phi(t,x)\in C_c^{\infty}\big([0,T)\times \R \big)$, 
where $p(x)=\frac12 e^{-\abs{x}}$ and the star $\star$ is the standard 
convolution product on $\R$ (for other details see \cite{gui2013wave}).  
\end{definition} 
We will test this definition on an example of a non-smooth  solution obtained from the peakon ansatz \cite{CH,qiao2012integrable,gui2013wave}, that is, we assume 
\begin{equation} \label{eq:peakonansatz} 
    u=\sum_{j=1}^n m_j (t)e^{-\abs{x-x_j(t)}}, 
\end{equation} 
where all $m_j$s are positive, and hence $ m=2\sum_{j=1}^n m_j \delta_{x_j}$ is a positive discrete measure.  We then substitute this ansatz into Definition  \ref{def:weakgui} to determine the $t$-dependence of $m_j(t), x_j(t)$.  
The relevant computation is done in \cite{gui2013wave} and we record 
the result following the elegant presentation in \cite{qiao2012integrable}: 
\begin{equation} \label{eq:peakonodesqiao}
    \dot m_i=0, \qquad \dot x_i=-\frac13m_i^2+\sum_{j,k}^n m_j m_k \big(1-\
    \sgn(x_i-x_j) \sgn(x_i-x_k)\big)e^{-\abs{x_j-x_i}-\abs{x_i-x_k}}, 
\end{equation} 
with the convention $\sgn(0)=0$.  We can specialize this expression to the set satisfying the 
ordering conditions $x_1<x_2<\dots<x_n$, in which case \eqref{eq:peakonodesqiao}
simplifies to: 
\begin{equation}\label{eq:peakonodesgui}
\dot m_j=0, \qquad \dot x_j=\frac23m_j^2+2m_j\sum_{i\neq j} m_i 
e^{-\abs{x_i-x_j}}+ 4
\sum_{i<j< k}^n m_i m_k e^{-\abs{x_i-x_k}}.  
\end{equation}
\begin{remark} 
We believe there is a slight misprint in \cite{gui2013wave} (see p. 22 therein), namely, 
the second term should have only non-diagonal summation as displayed above.  
We also note that stability of peakons satisfying \eqref{eq:peakonodesgui} 
was established in \cite{liu-liu-qu} and \cite{liu-liu-qu-multip}. 
\end{remark}

It is sufficient for our purposes to assume the number of peakons $n$ to be two.  
Thus 
\begin{equation*}
u=m_1(t) e^{-|x-x_1(t)|}+m_2(t) e^{-|x-x_2(t)|}, 
\end{equation*}
and the above definition of a weak solution in this special case 
results in the following system of ODEs: 
\begin{subequations}\label{eq:twopeakonsgui}
\begin{align} 
\dot m_1&=0=\dot m_2, \\
\dot x_1&=\frac23 m_1^2+2 m_1 m_2 e^{-\abs{x_1-x_2}}, \\
\dot x_2&=\frac23 m_2^2+2 m_1 m_2 e^{-\abs{x_1-x_2}}.  
\end{align}
\end{subequations}
The Sobolev norm of $u$ for the two-peakon ansatz can be easily computed (see 
the first part of the proof  of Theorem \ref{thm:H1conservation} for the general case) to be 
\begin{equation*}
||u||_{H^1}^2=2m_1u(x_1)+2m_2u(x_2)=2(m_1^2+m_2^2)+4m_1m_2 e^{-\abs{x_1-x_2}}, 
\end{equation*}
and this expression is clearly not $t$ invariant for general $m_1,m_2$ as one can easily see by 
computing $\frac{d}{dt} ||u||_{H^1}^2$ in the region $x_1<x_2$ with the help of 
\eqref{eq:twopeakonsgui}, in which case 
one obtains 
\begin{equation*}
\frac{d}{dt} ||u||_{H^1}^2=4m_1m_2(\dot x_1-\dot x_2)e^{x_1-x_2}=\frac{8m_1m_2}{3}(m_1^2-m_2^2)e^{x_1-x_2}, 
\end{equation*}
implying that the Sobolev $H^1$ norm of $u$ is only conserved in the case of 
$m_1=m_2$.  
\begin{remark} 
The concept of a weak solution proposed in \cite{gui2013wave} is 
a natural generalization of earlier definitions of weak solutions 
developed for nonlinear PDEs of Burgers' type, 
in other words inspired by one dimensional conservation laws (see \cite{constantin-escher-globalweakCH, constantin-escher } for a relevant discussion of  this type of weak solutions to CH; for other types of weak solutions to 
CH see \cite{constantin-bressan}) and it works well for certain fundamental questions, like general
existence theorems.  We are only pointing out that, perhaps essential, 
aspects of integrability known in the smooth sector may not survive the transition 
to the weak sector so defined. For other relevant work related 
to \eqref{eq:m1CH} and its weak solutions the reader is 
asked to consult  \cite{reyes-gorka-bies}.  
\end{remark}  

\subsection{ Lax integrability and peakons} 
What we are proposing in this paper is 
to preserve Lax integrability instead and let Lax integrability dictate the suitable 
paradigm for weak/distributional solutions of \eqref{eq:m1CH}.   First we state 
the final result in the case of the peakon ansatz, leaving the details to the later part of the paper.  In order to view 
\eqref{eq:m1CH} as a distributional equation involving a discrete measure $m$ one needs to define the product $u_x^2 m$.  We will argue below that the only 
choice consistent with Lax integrability is to take $u_x^2 m$ to mean 
\begin{equation} \label{eq:defprod} 
u_x^2 m\stackrel{def}{=}\langle u_x^2 \rangle m, 
\end{equation}
where $\langle f \rangle$ denotes the average function (the arithmetic average 
of the right hand and left hand limits). Since for the peakon ansatz
\begin{equation}
\dot m=2\sum_{j=1}^n \dot m_j \delta_{x_j}-2\sum_{j=1}^n m_j \dot x_j \delta'_{x_j},  
\end{equation}
hence \eqref{eq:m1CH}, with the rule \eqref{eq:defprod} in force,  
readily reduces to the system of ODEs: 
\begin{equation}\label{eq:xmODE}
\dot m_j=0, \qquad \dot x_j=u^2(x_j)-\langle u_x^2 \rangle(x_j).  
\end{equation}
Furthermore, assuming the ordering condition $x_1< x_2<\cdots<x_n$, we obtain 
another, more explicit form of \eqref{eq:xmODE}, namely
\begin{equation}\label{mCH_ode}
\dot{m}_j=0, \qquad 
\dot{x}_j=2\sum_{\substack{1\leq k\leq n,\\k\neq j}}m_jm_ke^{-|x_j-x_k|}+4\sum_{1\leq i<j<k\leq n}m_im_ke^{-|x_i-x_k|}\, .  
\end{equation}
We note that this system differs from \eqref{eq:peakonodesgui}; 
the difference amounting to the absence of the term 
$\frac23 m_j^2$ whose presence on the other hand can be traced back precisely to the definition of the singular product $u_x^2 m$.  Indeed, if we used Definition \ref{def:weakgui} then 
\begin{equation*}
u_x^2 m\stackrel{def\,\,  \text{\ref{def:weakgui}}}{=} \big(\frac{\avg{u_x^2}+2\avg{u_x}^2}{3}\big) m. 
 \end{equation*}  
 \subsection{Lax integrability and the preservation of the $H^1$ norm} 
 
 The following theorem supports the idea of building the concept of 
 weak solutions to \eqref{eq:m1CH} based on the multiplication 
 formula \eqref{eq:defprod}.  We have 
 
\begin{theorem} \label{thm:H1conservation} 
Let $u$ be given by the peakon ansatz \eqref{eq:peakonansatz} and let 
the multiplication of the singular term $u_x^2m$ in 
\eqref{eq:m1CH} be defined by \eqref{eq:defprod}.  Then if $u$ satisfies 
\eqref{eq:m1CH} and 

\begin{equation*} 
    \frac{d}{dt} ||u||_{H^1}=0.  
\end{equation*}
\end{theorem}
\begin{proof} 
First we establish that for the peakon ansatz the relation (see \eqref{eq:H1norm}) 
\begin{equation*}
||u||_{H^1}^2=\int um dx=\sum_{j=1}^n 2m_j u(x_j)
\end{equation*}
persists.  Indeed, setting $x_0=-\infty, x_{n+1}=\infty$ and partitioning $\R$ as $\cup_{j=0}^{n} (x_j, x_{j+1})$, we can write 
\begin{align*}
&||u||_{H^1}^2=\int \big(u^2+u_x^2\big)\, dx=\sum_{j=0}^n \int_{x_j+}^{x_{j+1}-} \big(u^2+u_x^2\big)\, dx=\\
&\int u^2\, dx + \sum_{j=0}^n uu_x\bigg|_{x_j+}^{x_{j+1}-}-
\sum_{j=0}^n \int_{x_j+}^{x_{j+1}-}uu_{xx} \, dx=
\sum_{j=0}^n uu_x\bigg|_{x_j+}^{x_{j+1}-}, 
\end{align*}
where in the last step we used that $u=u_{xx}$ holds away from the support 
of $m$.  If we denote by $[f](x_j)$ the jump of a piecewise continuous function 
$f$ at $x_j$ then we can write 
\begin{equation*}
||u||_{H^1}^2=-\sum_{j=1}^n \jump{u_x}(x_j) u(x_j)=\sum_{j=1}^n 2m_j u(x_j)=
\int u m \, dx, 
\end{equation*}
thus proving the claim.  
We proceed now to compute the time derivative of $||u||_{H^1}^2$.  We will 
carry out the computation in two steps.  
First we observe that as long as all $x_j$s are distinct, $u(x_j)$ is differentiable in $t$ and the derivative 
is given by the formula
\begin{equation*}
\frac{d}{dt}u(x_j)=\avg{u_x}(x_j)\dot x_j+\avg{u_t}(x_j), 
\end{equation*}
which in conjunction with \eqref{eq:peakonansatz}, \eqref{eq:defprod} and \eqref{mCH_ode} implies
\begin{align*}
&\frac{d}{dt}||u||_{H^1}^2=\frac{d}{dt} \big(\sum_{j=1}^n 2m_j u(x_j)\big)=
\sum_{j=1}^n 4m_j \avg{u_x}(x_j)\dot x_j\stackrel{\eqref{eq:xmODE}}{=}\sum_{j=1}^n 4m_j \avg{u_x}(x_j)\big(
u^2-\avg{u_x^2}\big)(x_j)=\\&-2\sum_{j=1}^n \jump{u_x}(x_j)\avg{u_x}(x_j)\big(
u^2-\avg{u_x^2}\big)(x_j)=-\sum_{j=1}^n \jump{u_x^2}(x_j)\big(
u^2-\avg{u_x^2}\big)(x_j)=\\&\sum_{j=1}^n \jump{u^2 -u_x^2}(x_j)\avg{
u^2-u_x^2}(x_j)=\frac12 \sum_{j=1}^n \jump{\big(u^2-u_x^2\big)^2}(x_j), 
\end{align*}
where we used multiple times the identity $\avg{f}(x_j)\jump{f}(x_j)=\frac12\jump{f^2}(x_j)$ valid for any piecewise continuous function $f$.  This establishes the 
identity 
\begin{equation*}
\frac{d}{dt}||u||_{H^1}^2=\frac12 \sum_{j=1}^n \jump{\big(u^2-u_x^2\big)^2}(x_j), 
\end{equation*}
which completes the first step of the computation.  
In the second step of the computation 
we note that for the peakon ansatz \eqref{eq:peakonansatz} $u^2-u_x^2$ is a 
piecewise constant function since on each interval $(x_j, x_{j+1})$ 
\begin{equation*}
u=A_j e^x+B_j e^{-x}, \quad u_x =A_je^x-B_je^{-x}, 
\end{equation*}
hence $u^2-u_x^2=4A_jB_j$ there.  Moreover, on the interval $(x_0, x_1)$, $B_0=0$, 
and on $(x_n, x_{n+1})$, $A_{n}=0$ respectively, implying in each case 
that $u^2-u_x^2=0$ on these two intervals.  In the final step of the proof 
we use that $u^2-u_x^2$ is a piecewise constant function to note that 
$\sum_{j=1}^n \jump{\big(u^2-u_x^2\big)^2}(x_j)$ is therefore a telescoping 
sum and hence, as the result of cancellations of interior terms and the absence of 
boundary terms, 
\begin{equation*}
\frac{d}{dt}||u||_{H^1}^2=\frac12 \big( (u^2-u_x^2)^2(x_n+)-(u^2-u_x^2)^2(x_1-))=0, 
\end{equation*}
which completes the proof.  
\end{proof} 
\section{Lax integrability and weak solutions} 
Here we provide some details as to why the originally ill-defined 
term $u_x^2 m$ in \eqref{eq:m1CH} should be regularized by 
using formula \eqref{eq:defprod} if  
Lax integrability is to be preserved.  We first recall that in the smooth sector the Lax 
pair for \eqref{eq:m1CH} reads \cite{qiao2006new}: 
\begin{equation}\label{eq:dLax-pair}
\Psi_x=\frac12U \Psi, \quad  \Psi _t =\frac12 V \Psi, \quad  \Psi=\begin{bmatrix} \Psi_1\\\Psi_2 \end{bmatrix}, 
\end{equation} 
with 
\begin{equation*} 
U=\begin{bmatrix} -1 &\lambda m\\ -\lambda m& 1 \end{bmatrix}, \qquad  
V=\begin{bmatrix} 4\lambda^{-2} + Q & -2\lambda^{-1} (u-u_x)-\lambda m Q\\
2\lambda^{-2}(u+u_x)+\lambda m Q & -Q \end{bmatrix}, \quad Q=u^2-u_x^2. 
\end{equation*} 
Furthermore, in the smooth sector, Lax integrability is understood to mean that the compatibility condition $\Psi_{xt}=\Psi_{tx}$ implies \eqref{eq:m1CH}.  
To promote \eqref{eq:dLax-pair} to a distributional Lax pair we need 
to ensure that all sides of \eqref{eq:dLax-pair} are 
well defined as distributions.  In our approach we 
consider $\Psi(\bullet,t)$ as a vector valued ($\R ^2$) distribution in $\mathcal{D}'(\R_x)$ and 
$D_x$ is the standard distributional derivative (in $x$).  We view 
the $t$ variable as a deformation parameter, in agreement with 
the original formulation by P. Lax \cite{lax-kdv},  which prompts us to view 
$\Psi(\bullet,t)$ as a differentiable map. 
\begin{definition}\label{def:PsiDist}
Let $T>0$ be given then $\Psi(\bullet, t)$ is a differentiable map
\begin{align*}
&\Psi(\bullet,t): (0, T)\ni t \longrightarrow \R^2\otimes\mathcal{D}'(\R_x), \qquad 
\text{where the  distributional derivative (with respect to $t$) is given by}\\& D_t\Psi(\bullet,t)\stackrel{def}{=} \lim_{h\rightarrow 0} \frac{ \Psi(\bullet, t+h)-\Psi(\bullet, t)}{h}, 
\end{align*}
provided the limit exists in $\R^2\otimes\mathcal{D}'(\R_x)$.  
\end{definition} 
To complete the definition of a distributional Lax pair we 
examine the conditions on the right hands sides of \eqref{eq:dLax-pair} needed to 
ensure that these are well defined distributions. In this letter we 
restrict our attention to the peakon ansatz \eqref{eq:peakonansatz}.  Away from 
the singular support of $m$, that is, away from the points $\{x_1,x_2, \cdots, x_n\}$, 
the distributional Lax pair will be smooth and given by \eqref{eq:dLax-pair}.  
On the singular support the multipliers of $m$ are not continuous functions, because 
neither is $\Psi$ as a piecewise smooth function, and nor is $Q=u^2-u_x^2$, being a piecewise constant function (see the proof of Theorem \ref{thm:H1conservation}), both with jumps on the singular support of $m$.  
Thus none of the multiplications $\Psi m, Q \Psi m$ is defined.  To define them 
we will have to assign values to $\Psi$ and $Q$ at the points $\{x_1,x_2, \cdots, x_n\}$. We postulate that the values of $\Psi$ at these points are a linear combination
of their respective left hand and right hand limits.  More precisely
\begin{definition} \label{def:invariantreg}
An invariant regularization of 
the Lax pair \eqref{eq:dLax-pair} valid for the peakon ansatz \eqref{eq:peakonansatz} is given by specifying 
the values of $\alpha, \beta \in \R$, subsequently setting 
\begin{equation*} 
\Psi(x_k, t)=\alpha \jump{\Psi}(x_k,t)+\beta\avg{\Psi}(x_k,t), 
\end{equation*}
and assigning some values $Q(x_k,t)$ to $Q(x,t)=u^2 -u_x^2$ 
at the points $\{x_1,x_2, \cdots, x_n\}$.  Then 
\begin{equation*}
\begin{aligned}
\Psi(\bullet, t)\delta _{x_k}&\stackrel{def}{=}\Psi(x_k,t) \delta_{x_k},   \qquad 1\leq k\leq n, \\
Q(\bullet,t)\delta_{x_k}&\stackrel{def}{=}Q(x_k,t) \delta_{x_k}, \qquad 1\leq k\leq n.  
\end{aligned}
\end{equation*}
\end{definition}
Now the distributional Lax pair 
\begin{equation}\label{eq:DLax-pair}
D_x\Psi=\frac12U \Psi, \quad  D_t\Psi =\frac12 V \Psi, \quad  \Psi=\begin{bmatrix} \Psi_1\\\Psi_2 \end{bmatrix}, 
\end{equation}
is well defined and we can ask meaningfully the question of 
compatibility, remembering that by Definition \ref{def:PsiDist} $D_t$ and $D_x$ indeed commute in action on $\Psi$.  We have 
\begin{theorem}\label{thm:invreg}
Let $m$ be the discrete measure associated to $u$ defined by \eqref{eq:peakonansatz}.
Given an invariant regularization in the sense of Definition \ref{def:invariantreg} the 
distributional Lax pair \eqref{eq:DLax-pair} is compatible, i.e. $D_tD_x\Psi=D_xD_t\Psi$,  if and only if 
the following conditions hold: 
\begin{subequations}
\begin{align} 
\beta =1, \qquad \qquad   \alpha^2&=\frac14, &
Q(x_k)&=\avg{Q}(x_k), &\label{eq:abQ}\\
\dot m_k&=0, &
\dot x_k&=Q(x_k)&. \label{eq:mx}
\end{align}
\end{subequations}

\end{theorem}
We note that \eqref{eq:mx} is precisely \eqref{eq:xmODE} whereas 
the choice of $\beta=1$, $\alpha=\frac12$ and $\alpha=-\frac12$ amounts to 
the choice of right limits, respectively left limits,  in the definition of 
an invariant regularization \ref{def:invariantreg}.  

For a proof of Theorem \ref{thm:invreg} as well as a discussion of the origin of the concept of an invariant regularization we refer the reader to our longer paper 
\cite{chang-szmigielski-m1CHlong} in which we give a complete solution to the \textbf{integrable} peakon problem \eqref{mCH_ode}; the solution is obtained by
analyzing associated boundary value problems 
for the distributional Lax pairs \eqref{eq:DLax-pair} with $\beta=1, \alpha=\pm \frac12$.

\section*{Acknowledgements}
The first author was supported in part by the Natural Sciences and Engineering Research Council of Canada (NSERC), the Department of Mathematics and
Statistics of the University of Saskatchewan, the Pacific Institute of the Mathematical Sciences (PIMS) through the PIMS postdoctoral fellowship, and by LSEC, Institute of Computational Mathematics and Scientific Engineering Computing, AMSS, Chinese Academy of Sciences. J.S. was supported in part by NSERC \#163953.

\bibliographystyle{abbrv}
 \def\cydot{\leavevmode\raise.4ex\hbox{.}}
  \def\cydot{\leavevmode\raise.4ex\hbox{.}} \def\cprime{$'$}

  \end{document}